\documentclass [journal]{IEEEtran}
\usepackage{graphicx}
\usepackage{subfigure}
\usepackage{caption}
\usepackage{subcaption}
\usepackage{stmaryrd}
\usepackage{verbatim}
\usepackage{amssymb}
\usepackage{amsmath}
\usepackage{amsfonts}
\usepackage{mathrsfs}
\usepackage{url}
\usepackage{mathabx}
\usepackage{wasysym}
\usepackage{amsthm}
\usepackage{cite}
\usepackage{fixltx2e}
\usepackage{float}
\usepackage{dblfloatfix}
\usepackage[ampersand]{easylist}
\usepackage[ruled]{algorithm}
\usepackage{algorithmic}
\usepackage{afterpage}
\usepackage{multirow}
\usepackage{xcolor}
\usepackage{upgreek}

\DeclareMathOperator*{\argmin}{arg\,min}

\newcommand\abs[1]{\left\lvert#1\right\rvert}

\newtheorem{lemma}{Lemma}
\newtheorem{remark}{Remark}

\begin{document}

\title{
Theoretical and Experimental Assessment of Large Beam Codebook at mmWave Devices: How Much is Enough?
\thanks{

Bora Bozkurt, Hasan Atalay Günel, Ali Görçin and İbrahim Hökelek are with the Communications and Signal Processing Research (HİSAR) Lab., T{\"{U}}B{\.{I}}TAK B{\.{I}}LGEM, Kocaeli 41470, Turkey. They are also with the Department of Electronics and Telecommunications Engineering
Istanbul Technical University, Istanbul 34467, Turkey (e-mails:\{bora.bozkurt, hasan.gunel, ali.gorcin, ibrahim.hokelek\}@tubitak.gov.tr).

Mohaned Chraiti is with the Department of Electronics Engineering, Sabanci University, {\.{I}}stanbul, Turkey (e-mail: mohaned.chraiti@sabanciuniv.edu)

Ali Ghrayeb is with College of Science and Engineering, Hamad Bin Khalifa University, Doha, Qatar (e-mail: aghrayeb@hbku.edu.qa)
}
}

\author{\IEEEauthorblockN{Bora Bozkurt, Hasan Atalay Günel, Mohaned Chraiti, İbrahim Hökelek, Ali Görçin, and Ali Ghrayeb}




}

\maketitle




\begin{abstract}
Modern millimeter wave (mmWave) transceivers come with a large number of antennas, each of which can support thousands of phase shifter configurations. This capability enables beam sweeping with fine angular resolution, but results in large codebook sizes that can span more than six orders of magnitude. On the other hand, the mobility of user terminals and their randomly changing orientations require constantly adjusting the beam direction. A key focus of recent research has been on the design of beam sweeping codebooks that balance a trade-off between the achievable gain and the beam search time, governed by the codebook size. In this paper, we investigate the extent to which a large codebook can be reduced to fewer steering vectors while covering the entire angular space and maintaining performance close to the maximum array gain. We derive a closed-form expression for the angular coverage range of a steering vector, subject to maintaining a gain loss within \(\gamma\) dB (e.g., 2\, dB) with respect to the maximum gain achieved by an infinitely large codebook. We demonstrate, both theoretically and experimentally, that a large beam-steering codebooks (such as the \(1024^{16}\) set considered in our experiment) can be reduced to just a few steering vectors. This framework serves as a proof that only a few steering vectors are sufficient to achieve near-maximum gain, challenging the common belief that a large codebook with fine angular resolution is essential to fully reap the benefits of an antenna array.

\end{abstract}


\section{Introduction}
\subsection{Motivations}
The increasing demand for gigabit-per-second wireless links and massive connectivity drives the wireless industry and academic exploration into millimeter-wave (mmWave) and higher frequency bands, which offer substantial gigahertz-level bandwidth. Transmission over mmWave bands, however, comes at the expense of sensitivity to blockage and severe path loss. Therefore, deploying large antenna arrays at the users' terminals is deemed necessary to achieve high gains \cite{milimeter2018hemadeh}. Considerable effort has been directed towards innovative beam steering techniques, focusing on reducing the time required to search for optimal beams without compromising the gains from antenna arrays \cite{statisticsaided2018lin, statistical2017khormuji, datadriven2022ozkoc,hierarchical2016xiao, multi2017noh, enhanced2018xiao}.

\subsection{Related Work}
The high cost and power consumption of mmWave radio-frequency (RF) chains favor transceiver designs with a small number of RF chains and large antenna arrays \cite{milimeter2018hemadeh, beamforming2016kutty}. Consequently, analog beamforming (referred to here as beam-steering) emerged as a viable approach, in which the beam is steered by adjusting phase shifters, utilizing as few as one RF chain \cite{spatially2014ayach, channel2014alkhateeb, statistically2020shaban}. 

Numerous efforts have been made to develop innovative approaches to devise the phase shifter configurations and maximize the array gain while minimizing the beam search time and overhead. These methods can be classified into two categories: channel state information (CSI)-based beam-steering \cite{spatially2014ayach,channel2014alkhateeb} and codebook-based beam-steering \cite{statistical2017khormuji, statisticsaided2018lin, datadriven2022ozkoc,hierarchical2016xiao, multi2017noh, enhanced2018xiao}.

\subsubsection{CSI-based beam-steering}
The CSI-based approach resembles its counterpart in digital beamforming, as it requires estimating the CSI for each antenna. The key distinction, however, lies in selecting the antenna configuration from a discrete set of potential phase-shifter configurations. Under the assumption of perfect CSI knowledge, the authors in \cite{spatially2014ayach,channel2014alkhateeb} showed that near-maximum array gain can be achieved. 

In \cite{spatially2014ayach}, the authors formulate the beam search as a sparse reconstruction problem. They introduce a low-complexity algorithm for large antenna arrays, drawing inspiration from the well-known orthogonal matching pursuit technique. The work in \cite{channel2014alkhateeb} relies on the sparsity of the mmWave channel matrix to develop a low-complexity channel estimation algorithm for beamforming purposes. Although these techniques demonstrate optimal performance, they rely on the assumption of having CSI knowledge for each antenna array's element. However, factors such as short channel coherence time (resulting in rapid channel variation), high path loss (weak channel before beamforming), and deployment of a large number of antennas require high pilot symbol rates to obtain viable CSI estimates \cite{milimeter2018hemadeh, beamforming2016kutty}. For instance, transmitting over the 30\,GHz bands reduces the channel coherence time by a factor of 15 compared to transmitting over the 2\,GHz bands, suggesting an increase in overhead of at least 15 times.

\subsubsection{Codebook based beam-steering}
To eliminate the need for CSI knowledge, beam sweeping emerges as a viable alternative, where a beam is sequentially steered in different directions based on a predefined codebook. The direction that maximizes the array gain is then selected. This mechanism is supported by the 3GPP standard \cite{standard}. However, modern mmWave systems feature large antenna arrays and a high number of phase shifter configurations, enabling fine angular resolution for beam sweeping. These result in large codebooks, making exhaustive search time-inefficient. With a common belief that achieving near-optimal gain requires sweeping with small steps, various codebook designs have been proposed to reduce the angular search space by leveraging side information and intelligent algorithms \cite{fast2018filippini, mmwave2019wang, statisticsaided2018lin, statistical2017khormuji, datadriven2022ozkoc, beamcodebook2019mo, codebook2022mabrouki,hierarchical2016xiao, multi2017noh, enhanced2018xiao}.

\par Codebook design techniques assisted by statistical data have been put forward \cite{fast2018filippini, mmwave2019wang, statisticsaided2018lin, statistical2017khormuji, datadriven2022ozkoc}. The authors in\cite{statisticsaided2018lin} suggest collecting statistics about the Angle-of-Departure (AoD), assuming that certain angles are more probable than others. They utilize the Vector Quantization technique to devise a set of the most probable AoDs that represent the entire set. The work in \cite{statistical2017khormuji} considers the scenario of multiple-antenna base station (BS) and single-antenna users. In line with the approach in \cite{statisticsaided2018lin}, the authors assume that the AoDs are not uniformly distributed at the BS. They subsequently develop an AoD quantization technique and a non-uniform codebook, in which certain steering vectors are more probable than others. The authors propose allocating power according to the probability of AoD to maximize coverage. The work in \cite{datadriven2022ozkoc} investigates the potential construction of a data-driven codebook that maximizes coverage probability, considering a multiple-antenna BS, single-antenna receivers, and a predetermined codebook size. The algorithm iteratively adds new beam-steering codewords (phase shifters' configurations) to incrementally accommodate a new set of users until the pre-determined codebook size is reached.

Leveraging side information, such as AoD distribution statistics and location, can reduce the angular search space and consequently the codebook size. In addition to the fact that these approaches are designed for beam steering at the BS, experimental results demonstrate that they are eventually more suitable to be applied at the BS side, where statistics are collected based on user distribution \cite{Jchraiti1,Cchraiti}. Conversely, on the user device side, its random orientation results in nearly equiprobable DoAs over time, making the assumption less reliable. Moreover, even if certain directions with respect to a predefined reference direction are more probable in a given environment, an orientation estimate of the device is necessary for beam alignment. For example, in the case of line-of-sight scenario and where the user’s location is known, an orientation estimate of the user device is essential for determining the beam direction. Fingerprinting technique is another example of statistics-based techniques that uses location and orientation as side information. Despite their time search efficiency, experimental results indicate a significant gap between the achievable gain and maximum possible gain, reaching up to 17 dB \cite{Cchraiti}. 

Without considering any assumption on the distribution of AoD, hierarchical codebooks cover the entire angular space by intelligently narrowing down the search space\cite{hierarchical2016xiao, multi2017noh}. These schemes employ a multi-layered codebook to progressively refine beamwidth through successive stages. Hierarchical methods employ a binary search to find the optimal beam direction, starting with low-resolution (coarse) beams and advancing to high-resolution (narrow) beams. Although the hierarchical codebook schemes offer time-efficient solutions, their efficacy is hindered by the high probability of selecting the direction pointing to a non-dominant multi-path component, resulting in a high gain gap from the maximum \cite{hierarchical2016xiao,Cchraiti}
\subsection{Problem Statement} 
Statistics-based codebooks reduce the search space and provide a promising approach to overcome the need for CSI in beamforming and have demonstrated superior performance over hierarchical techniques in terms of array gain. 
Under the assumption of certain directions being more likely, results indicate their efficiency when implemented at the base station (BS), enhancing coverage and aggregated transmission rates. Such an assumption holds true at BSs, as they are static, but may not be valid for mobile user terminals due to their dynamically changing orientations. With the user terminal capable of covering the entire angular range, a key question arises: What should be the minimum codebook size (beam sweeping step) that ensures full angular coverage without significantly compromising the achievable array gain? We aim to answer this question through a theoretical framework and experimental validation. 
\subsection{Contributions} 
We consider the problem of beam sweeping at user terminals with random orientations. Without any assumption regarding the DoA distribution or orientation estimate, we examine whether large codebooks can be reduced to a few elements while maintaining the array gain within \(\gamma\) dB of the maximum, i.e., achieving near-optimal performance with only a few steering vectors. We project that a steering vector pointing a beam at a given angle can cover neighboring angles while maintaining an array gain within the $\gamma$\,dB margin from the maximum. We theoretically quantify the extent to which a steering vector can cover adjacent angles, showing that this coverage depends on both the angles and the array size. Building on the analytical results, we develop an algorithm that identifies the set of steering vector candidates, forming a refined codebook. Theoretical and experimental results show that it is possible to reduce large size codebook, to a few steering vector. For instance, experimental and theoretical results indicate that it is possible to maintain the gap within 3 dB of the maximum gain while reducing the codebook size from $1024^{16}$ to fewer than ten elements for the case of an Uniform Linear Array.

Our contributions can be summarized as follows.
\begin{itemize}
    \item We provide a framework to quantify the angular coverage of steering vector as a function of a constrained gain loss.
    \item We derive the refined codebook size and beam steering directions that ensure a gain loss not exceeding \(\gamma\) dB compared to the maximum gain obtained using an infinitely large codebook with fine angular resolution.
    \item We provide intricate insights into codebook reduction, revealing that the sweeping step should not be uniform in all directions. Moreover, contrary to the common belief that achieving near maximum gain requires small steps, our findings show that a small-sized codebook with coarse steps is enough to achieve near optimal performance. Beyond a certain size, the gain becomes negligible.
    \item We validate the proposed framework through experimental results using a multi-antenna mobile receiver, with an initial codebook size spanning six orders of magnitude.
\end{itemize}

This work provides insights into beam sweeping at the user device. When a small number of steering vectors are sufficient to cover the entire angular space, there will be no need for side information to reduce the search space. Furthermore, if side information is available to constrain the search space, the proposed approach can be utilized to perform beam sweeping with coarse steps, further reducing the search time.

The  paper is organized as follows. Sec. II presents the system model. Sec. III provides the theoretical derivation of the angular coverage ranges of a steering vector. Sec. IV provides details on codebook refinement algorithm. Sec. V details the experimental results, including all pertinent information the measurement results in depth. Finally, Sec. VI offers a summary of the main findings and the implications of this study.

\vspace{-0.3cm}

\section{System Model}\label{sec:sysmodel} 
This work is supported by a real-world experiment. Consequently, the system model described in this section is an abstraction of the experimental setup detailed in Sec. \ref{sec:expSetup}. We consider the scenario of a mobile and randomly oriented receiver equipped with $N$ antennas. Each antenna in the array is controlled by an \( M \)-bit phase shifter, enabling \( 2^M \) possible settings per antenna. The set of possible steering vectors forms an initial codebook (denoted by $\zeta$) of size \( (2^M)^N \). 

The DoA varies as the terminal's orientation changes or as objects in the environment move. Communications occur in the mmWave band, known for its limited scattering, with beam sweeping primarily aimed at tracking the direction of the strongest arriving ray. In a given instance, the DoA of a ray is defined as the angle between the array's boresight (normal to the array's midpoint) and the wavefront normal. In the presence of multiple incident rays, whether in line-of-sight or non-line-of-sight conditions, beam sweeping sequentially steers beams from a predefined codebook and selects the direction that maximizes array gain, i.e., ideally pointing in the direction of the strongest ray. As shown in Fig. \ref{fig:beamsweeping}, a larger codebook (i.e., smaller step size) increases the likelihood of aligning the DoA of the strongest ray more closely with the central direction of a swept beam, thereby enhancing array gain.

The study encompasses two types of phased arrays: the Uniform Linear Array (ULA) and the Uniform Rectangular Array (URA). For a ULA, an incident ray is characterized by the azimuthal DoA \(\theta\). For the URA, the incident ray is defined by the azimuth-elevation angle pair \((\theta_1, \theta_2)\). For clarity and simplicity, we present the mathematical expressions and relevant studies on URA in Sec. \ref{DevFromMaxGainForURA}.

In the case of an ULA, the array manifold vector \(\mathbf{a}(\theta)\) corresponding to the strongest ray arriving from a DoA \(\theta\) is given by:
\begin{equation}\label{eq:1}
\mathbf{a}(\theta) = \beta[1, e^{j \frac{2 \pi}{\lambda} d \sin (\theta)}, \ldots, e^{j(N-1) \frac{2 \pi}{\lambda} d \sin (\theta)}]^{\mathrm{T}},
\end{equation}
where $d$ is the antenna inter-distance and $\lambda$ is the wavelength. Here, $\beta$ represent the path loss effect. 

In analog beam steering, the phase-shifters are adjusted to steer the signal in a direction that maximizes the array gain.
With an \( M \)-bit phase shifter, each antenna phase \( \phi_n \) (\( n \in 1, \dots, N \)) can take values from the discrete set \( \mathcal{\phi} = \left\{ 0, \frac{2\pi}{2^M}, \ldots, \frac{(2^M - 1)2\pi}{2^M} \right\} \). To achieve the maximum array gain, the elements of the steering vector \( [\mathbf{w}]_n =\frac{1}{\sqrt{N}} e^{\mathrm{j} \phi_n} \) should ideally match, or approximate closely, the conjugate of the corresponding elements in the array manifold vector. That is:
\[
[\mathbf{w}]_n \propto [\mathbf{a}(\theta)]^{*}_n,  \quad \text{where} \quad \phi_n \simeq \arg([\mathbf{a}(\theta)]^{*}_n)
\]
In the context of beam sweeping, large codebooks with fine granular angular resolution guarantee that one of the steering vectors closely aligns with the optimal phase profile of the strong ray, thereby achieving the maximum gain.

\begin{figure}[t]
\includegraphics[width = 1\linewidth]{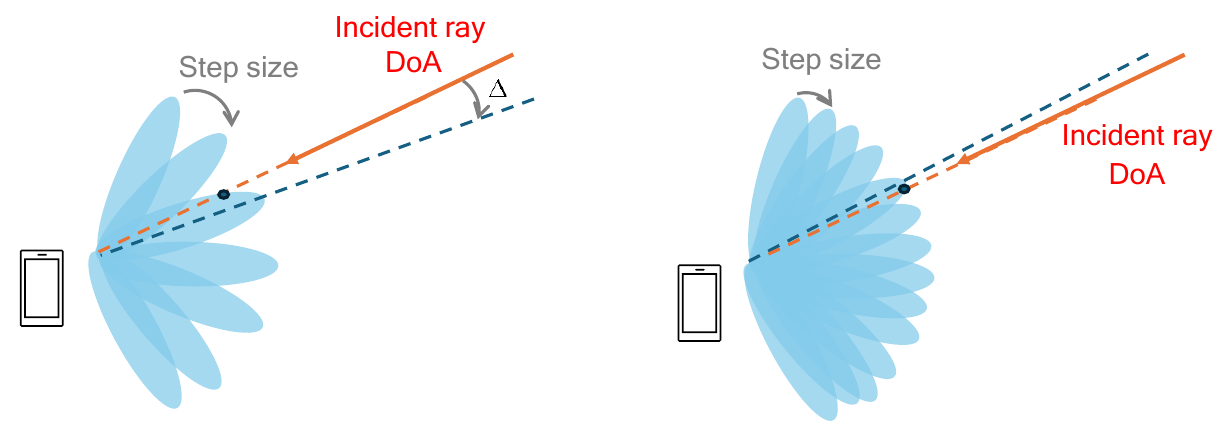}
\caption{Illustration of the effect of the beam sweeping step on the array gain.} \vspace{-0.3cm}
\centering
\label{fig:beamsweeping}
\end{figure}

As the number of antenna elements and phase configurations increases, the beam sweeping step is refined to achieve the maximum array gain. However, the codebook size, \((2^M)^N\), increases exponentially, leading to longer beam search times.\footnote{It is possible to find the phase-shifter configurations through channel estimation. However, this approach has multiple drawbacks, including high overhead, weak channel estimates before beamforming, and the requirement for channel estimation per antenna, making codebook-based beam steering a favorable option.} For example, in the considered experimental setup, a receiver with 16 antennas and 10-bit phase shifters results in a total of \( 1024^{16} \) possible steering vectors.

In light of the above discussion, we explore reducing the codebook size while covering the entire angular space with minimal loss of optimality, i.e., keeping the gain within \(\gamma\) dB of the maximum. We project that a steering vector pointing to a specific direction can also cover slightly deviated angles \(\theta + \Delta\), where \(\Delta\) is small. The notion of "small deviation" is defined with respect to the targeted array gain loss threshold $\gamma$. This analysis will enable us to determine the minimum number of steering vectors needed for full angular coverage while maintaining the array gain gap to maximum gain above the threshold \(\gamma\) dB.

\section{Angular coverage of a steering vector}
\subsection{The case of ULA}
Consider an ULA with a strong ray arriving from a DoA equal to \(\theta + \Delta\). Using the expression in \eqref{eq:1}, the manifold vector $\mathbf{a}(\theta+\Delta)$ is given by:
\begin{equation}\label{eq:arraymanifoldULA}
\mathbf{a}(\theta+\Delta) = \beta[1, e^{j \frac{2 \pi}{\lambda} d \sin (\theta+\Delta)}, \ldots, e^{j(N-1) \frac{2 \pi}{\lambda} d \sin (\theta+\Delta)}]^{\mathrm{T}}.
\end{equation} For a beam-steering vector $$\mathbf{w}=\frac{1}{\sqrt{N}}[e^{-j\phi_1}, e^{-j\phi_2},..., e^{-j\phi_N}],$$ the antenna array gain can be written as
\begin{equation}
    G(\mathbf{w},\theta+\Delta)= \abs{\mathbf{a}(\theta+\Delta)\cdot\mathbf{w}}^2 = \frac{\abs{\beta}^2}{N}\abs{\sum_{n=1}^{N}a_n(\theta)e^{-j\phi_n}}^2. 
\end{equation}
 Let \(\mathbf{w}_{\theta+\Delta}\) be the steering vector in \(\mathcal{C}\) that maximizes \(G(\mathbf{w}, \theta+\Delta)\) defined as:
  \begin{equation}
  \mathbf{w}_{\theta+\Delta} = \arg \max_{\mathbf{w} \in \mathcal{C}} G(\mathbf{w}, \theta+\Delta).
  \end{equation} 
  The latter Unit-Modulus Least Squares problem can be solved using techniques such as Semi-Definite Relaxation and associated algorithms \cite{fast2017tranter}, which are suitable for both moderate and large values of \(M\) (the phase shifter configuration set size per antenna.) Nonetheless, the maximum array gain can be achieved through infinity large codebook ($M$ is large) as in this work. The problem simplifies and the optimal solution involves setting each antenna phase shifter to 
  \begin{equation}\label{eq:ULAOptimalW}
  \begin{aligned}
  \arg([\mathbf{w}_{\theta+\Delta}]_n) &\approx \arg(a_n(\theta+\Delta)^*)\\
  &=-j(n-1) \frac{2 \pi}{\lambda} d \sin (\theta+\Delta),
  \end{aligned}\end{equation}
  
  yielding a maximum array gain equal to
  \begin{equation}
  G(\mathbf{w}_{\theta+\Delta}, \theta+\Delta) \approx N \beta^2. 
  \end{equation}
  Note that $G(\mathbf{w}_{\theta+\Delta},\theta+\Delta) \leq N \beta^2$ for small codebook size, with coarse step size.

Considering the DoA \( (\theta + \Delta) \) of the incident wave and applying the steering vector \[ \mathbf{w}_{\theta}=\frac{1}{N}[1, e^{-j \frac{2 \pi}{\lambda} d \sin (\theta)}, \ldots, e^{-j(N-1) \frac{2 \pi}{\lambda} d \sin (\theta)}] \] pointing toward a slightly different direction \( \theta \), the ratio of the achievable array gain to the maximum array gain is given by:

\begin{subequations}\label{eq:ULAdev}
\begin{align}
    D_\theta(\Delta) &= \frac{G(\mathbf{w}_\theta, \theta + \Delta)}{G(\mathbf{w}_{\theta+\Delta}, \theta+\Delta)}\\&= \frac{\abs{ \mathbf{a}(\theta+\Delta)\cdot\mathbf{w}_{\theta}}^2}{\abs{\mathbf{a}(\theta+\Delta)\cdot\mathbf{w}_{\theta+\Delta}}^2}\\& \geq \frac{\abs{\mathbf{a}(\theta+\Delta)\cdot\mathbf{w}_{\theta}}^2}{N\beta^2}.\label{eq:suba} 
    \end{align}
\end{subequations}
The inequality in \eqref{eq:suba} stems from the fact that $G(\mathbf{w}_{\theta+\Delta},\theta+\Delta) \leq N \beta^2$.
The lower bound becomes approximately an equality in the case when $M$ is large.
By incorporating the expressions for $\mathbf{a}(\theta+\Delta)$ and $\mathbf{w}_{\theta}$, given in \eqref{eq:arraymanifoldULA} and \eqref{eq:ULAOptimalW}, respectively, we obtain
\begin{equation}
\begin{aligned}
    D_\theta(\Delta) &\geq 
    \frac{\frac{1}{N}\abs{\sum_{n=1}^{N}e^{j(n-1) \frac{2 \pi}{\lambda} d (\sin (\theta + \Delta) - \sin(\theta))}}^2}{N}\\
    &=\frac{1}{N^2}\abs{\sum_{n=1}^{N}e^{j(n-1) \frac{2 \pi}{\lambda} d (\sin (\theta + \Delta) - \sin(\theta))}}^2.
    \end{aligned}
\end{equation}
$D_\theta(\Delta)$ is equal to one, if and only if $\Delta=0$, suggesting that the achievable gain and the maximum gain are equal only in the case of perfect alignment between the indcident wave DoA and the direction of the beam. Using the geometric series related formula, the lower bound becomes
\begin{equation}
\begin{aligned}
    D_\theta(\Delta)&\geq \frac{1}{N^2}\abs{\frac{1-e^{jN\frac{2\pi}{\lambda}d(\sin(\theta+\Delta) - \sin(\theta))}}{1-e^{j\frac{2\pi}{\lambda}d(\sin(\theta+\Delta) - \sin(\theta))}}}^2 \\&
    = \frac{1}{N^2}\frac{1-\cos\left[N\frac{2\pi}{\lambda}d(\sin(\theta+\Delta)-\sin(\theta))\right]}{1-\cos\left[\frac{2\pi}{\lambda}d(\sin\theta+\Delta)-\sin(\theta))\right]}.
    \end{aligned}
\end{equation}
Let \( z = \frac{2\pi}{\lambda} d (\sin(\theta+\Delta) - \sin(\theta)) \) represents the phase deviation due to the misalignment between the actual DoA \( \theta + \Delta \) and the steering direction \( \theta \). Define the degradation factor threshold as \( \gamma_{\mathrm{f}} = 10^{\gamma/10} \), where \( \gamma \) represents the allowable degradation in dB, i.e., a \(\gamma\)-dB gap between the achievable and maximum array gain. The alignment condition $D_{\theta}(\Delta)\geq \frac{1}{\gamma_f}$ ensuring that the array gain loss does not exceed this threshold is achieved when:

\[\frac{1}{N^2}\frac{1-\cos\left(Nz\right)}{1-\cos\left(z\right)} \geq \frac{1}{\gamma_{\mathrm{f}}},\]
which gives
\begin{equation}\label{eq:inequalConj}
\begin{aligned}
     \cos(z) -\frac{\gamma_{\mathrm{f}}}{N^2}(\cos(Nz)-1) - 1  \geq 0    .
\end{aligned}
\end{equation}
Solving (\ref{eq:inequalConj}) yields a set of values for \( z \) that will help defining the coverage range of a steering vector. The solution set is provided in the following lemma.

\begin{lemma}\label{lem1}
Considering moderately large number of antennas $N$ (e.g., higher than six), there exists a constant \(\alpha^{*}\) such that the following inequality
\begin{equation}\label{eq:ineLem}
\cos(z) - \frac{\gamma}{N^2} (\cos(Nz) - 1) - 1  \geq 0
\end{equation}
holds for
$$
z \in \left[-\frac{\alpha^{*}}{N}, \frac{\alpha^{*}}{N}\right].
$$
$\alpha^{*}$ is the smallest strictly positive root of $$1-\cos(\alpha) - \frac{\alpha^2}{2\gamma}=0.$$
\end{lemma}
\begin{proof}
Consider the inequality
\begin{equation} \label{eq:mainEq}
  \cos(z) - \frac{\gamma}{N^2} (\cos(Nz) - 1) - 1 \geq 0. 
\end{equation}
Let us denote by \(\alpha\) the product \(zN\). The inequality becomes
$$
\cos\left(\frac{\alpha}{N}\right) - \frac{\gamma}{N^2} (\cos(\alpha) - 1) - 1 \geq 0,
$$
implying
\begin{equation}\label{eq:ineq}
-\overbrace{\frac{\frac{N^2}{\gamma}(1-\cos(\alpha/N))}{N^2/\gamma}}^{\textcircled{1}} - \frac{\gamma}{N^2} (\cos(\alpha) - 1) \geq 0.
\end{equation}
Incorporating L'Hôpital's rule for high and moderate values of \(N\) to the numerator of the term in \(\textcircled{1}\) gives
$$
\frac{N^2}{\gamma}(1 - \cos(\alpha/N)) \simeq \frac{\alpha^2}{2\gamma}.
$$
Consequently, the inequality in \eqref{eq:ineq} becomes
\begin{equation}
\frac{\gamma}{N^2}(1 - \cos(\alpha) - \frac{\alpha^2}{2\gamma}) \geq 0.
\end{equation}
Since \(\frac{\gamma}{N^2} > 0\), the equality simplifies to
\begin{equation}\label{ineqGeq0}
1 - \cos(\alpha) - \frac{\alpha^2}{2\gamma} \geq 0. 
\end{equation}
The roots are the points where the function crosses zero or is at least tangent to zero. To determine the ranges where the inequality holds, we begin by finding the roots of the equation:
\begin{equation}
1 - \cos(\alpha) - \frac{\alpha^2}{2\gamma} = 0.
\end{equation}
Since this equation involves only even functions, if \( \alpha \) is a root, then \( -\alpha \) must also be a root. The fact that zero is a root implies that the set of roots contains an odd number of elements, which has the following form:
\begin{equation}
\{-\alpha_K, \dots, -\alpha_1, 0, \alpha_1, \dots, \alpha_K\},
\end{equation}
where \( \alpha_K>\cdots>\alpha_2>\alpha_1 \).

For \( \alpha = 0 \), the first derivative:
\begin{equation}
\sin(\alpha) - \frac{\alpha}{\gamma_{\mathrm{f}}}
\end{equation}
is equal to zero, while the second derivative:
\begin{equation}
\cos(\alpha) - \frac{1}{\gamma_{\mathrm{f}}}=1-\frac{1}{\gamma_{\mathrm{f}}}
\end{equation}
is strictly positive, given that \( \gamma_{\mathrm{f}} > 1 \). This suggests that \( \alpha = 0 \) is a local minima. The other roots do not exhibit the properties of local minima or maxima, as satisfying the necessary condition for being a local extremum involves solving the following system of equations:
\begin{equation}
\begin{cases}
1 - \cos(\alpha) - \frac{\alpha^2}{2\gamma} = 0 \\
\sin(\alpha) - \frac{\alpha}{2\gamma} = 0
\end{cases}
\end{equation}
This equation system yields to:
\begin{equation}
1 - \cos(\alpha) - 2\gamma \sin^2(\alpha) = 0.
\end{equation}
The latter equation approaches zero only as \( \alpha \) approaches zero. Thus, none of the roots, except \( \alpha = 0 \), correspond to local minima. Consequently, the function changes sign at each root. Collectively, these results suggest that the inequality in \eqref{ineqGeq0} holds when $N z$ is within the following union intervals:
\[\begin{aligned}
&[-\alpha_1, \alpha_1] \cup [\alpha_2, \alpha_3] \cup [-\alpha_3,-\alpha_2] \cup \cdots\\&\qquad\qquad\cup [\alpha_{2k}, \alpha_{2k+1}] \cup [-\alpha_{2k+1}, -\alpha_{2k}],
\end{aligned}
\]
where \( k = \lfloor \frac{K}{2} \rfloor \). It follows that the function is positive within the interval \( [-\alpha_1, \alpha_1] \), centered around zero, as a part of the union. Here, $\alpha_1$ is the smallest strictly positive root.
\begin{remark}
The number of roots of the function $1-\cos(\alpha) - \frac{\alpha^2}{2\gamma}=0$ depends on the value of \( \gamma_{\mathrm{f}} > 1 \), with a minimum of three roots, including \( \alpha = 0 \). The proof is provided in the following. First, \( \alpha = 0 \) represents a local minima, at which the function equals zero. These indicates that around \( \alpha = 0 \) the function rises. Furthermore, as \( \alpha \to \pm \infty \), the function tends toward \( -\infty \). This behavior implies that the function must cross zero at least at two additional points, which are symmetrically opposite due to the even nature of the function. Therefore, there exits at least two non-zero roots between which the inequality in $\eqref{eq:ineLem}$ is valid. These roots are denoted by $\alpha_1$ and $-\alpha_1$.
Second, the function may have additional roots if it continuously alternates between increasing and decreasing. Therefore, we analyze the first derivative. The roots of the first derivative ($\sin(\alpha) - \frac{\alpha}{\gamma_{\mathrm f}}$) set to zero correspond to intersections between a line with slope \( 1 / \gamma_{\mathrm{f}} \) and the curve \( \sin(\alpha) \). Considering the case when $\alpha>0$ (similar analysis applies for $\alpha<0$). In the interval \( [0, \pi] \), it is well known that any line, regardless of the slope, intersects with the sine function at maximum in one point. When \( \alpha\in[\pi, 2\pi] \), the sine function is negative, precluding intersections. For $\alpha>2\pi$, intersections between $\sin(\alpha)$ and $\alpha / \gamma_{\mathrm{f}}$ may occur again, if \( \alpha / \gamma_{\mathrm{f}} < 1 \), as the sine function's maximum is bounded by one. Hence, no intersect occurs if the function is continuously increasing or decreasing. In the range \( \alpha > 2\pi \), no intersection occurs when \( \alpha / \gamma_{\mathrm{f}} > 1 \) corresponding to the case when \( \gamma_{\mathrm{f}} < 2\pi \) and hence in case $\gamma\leq 10\log_{10}(2\pi)\simeq 8$\,dB. 

In conclusion, with a maximum allowed degradation $\gamma$ of 8 dB or less, we obtain exactly three roots, including zero. The solution set of the inequality in (\ref{eq:ineLem}) is hence of the form \( [-\alpha_1, \alpha_1] \). For degradation values exceeding 8 dB, additional roots may exist. Nonetheless, the lemma holds, as the function remains positive within \( [-\alpha_1, \alpha_1] \).\end{remark}
\end{proof}

In Lemma \ref{lem1}, we provide a bound on \( z \) (and hence on \( \sin(\theta + \Delta) - \sin(\theta) \)) within which the array gain degradation remains within a predefined threshold. This bound is given by the interval \(\left[-\frac{\alpha^{*}}{N}, \frac{\alpha^{*}}{N}\right]\), where \(\alpha^{*}\) can be determined numerically by solving \(1 - \cos(\alpha) - \frac{\alpha^2}{2\gamma} = 0\). While this interval is derived considering the assumption of large number of antennas at the user device, numerical results show that the analytical range is closely approximation with negligible discrepancy, as demonstrated in Tab. \ref{tab:alphaApprox} for various values of \(\gamma\) and antennas number $N$. The table shows that the difference in size between the numerical solution and the analytical derived range is negligible. 
\begin{table}[h!]
\centering
\caption{Discrepancy between the numerical (actual) and approximated values of the solution range pertains $\alpha$.}
\label{tab:alphaApprox}
\begin{tabular}{|c|c|c|c|c|}
\hline
\textbf{$\gamma_{\mathrm f}$} & $N \geq 10$ & $N \geq 20$ & $N \geq 50$ & $N \geq 100$\\
\hline
$2$ & 0.0121 & 0.0030 & 4.8054e-04 & 1.2011e-04\\
$3$ & 0.0135 & 0.0034 & 5.3849e-04 & 1.3460e-04\\
$4$ & 0.0139 & 0.0035 & 5.5446e-04 & 1.3859e-04\\
$5$ & 0.0140 & 0.0035 & 5.5693e-04 & 1.3921e-04\\
\hline
\end{tabular}
\end{table}

Using the results from Lem. \ref{lem1}, which provide bounds for \( z = \sin(\theta + \Delta) - \sin(\theta) \) in the form \( \left[-\frac{\alpha^{*}}{N}, \frac{\alpha^{*}}{N}\right] \) that maintain the loss within \( \gamma \), we can subsequently determine the corresponding range for \( \Delta \). Lem. \ref{lem1} suggests that the inequality (\ref{eq:inequalConj}) holds if, 
\begin{equation}\label{eq:sintheta}
    -\frac{\alpha^{*}}{N}  \leq z=\frac{2\pi}{\lambda}d(\sin(\theta + \Delta)-\sin(\theta)) \leq \frac{\alpha^{*}}{N},
\end{equation}
which implies
\begin{equation}\label{eq:boundDelta}
    -\frac{\lambda}{2\pi d N} \alpha^{*}+\sin(\theta) \leq \sin(\theta + \Delta)\leq \frac{\lambda}{2\pi dN}\alpha^{*}+\sin(\theta).
\end{equation}
For a given steering angle \(\theta\), the coverage angular range of the steering vector is hence bounded as follows:
\[
\mathrm{L}\Delta_{\theta} \leq \Delta \leq \mathrm{U}\Delta_{\theta},
\]
where 
\begin{align}\label{eq:ULAdeltaInterval}&\qquad\mathrm{L}\Delta_{\theta}=\arcsin\left(-\frac{\lambda}{2\pi d N} \alpha^{*}+\sin(\theta)\right)-\theta\\
    &\text{and}\nonumber\\
    &\qquad\mathrm{U}\Delta_{\theta}= \arcsin\left(\frac{\lambda}{2\pi dN}\alpha^{*}+\sin(\theta)\right)-\theta.
\end{align}
\begin{remark} In order to be able to inverse the $sine$ function, a rule of thumb is that the boundary of $\theta+\Delta$ should be in the range [$-\pi/2, \pi/2$]. This is somewhat straightforward as it is a common knowledge that the visibility region (effective coverage range) of an ULA is between [$-\pi/2, \pi/2$] and hence $\theta+\Delta$ should be \cite{wide2023ming}. Moreover, in practice, the visible region is normally smaller to avoid granting lobe governed by the following constraint:

\begin{equation} \label{eq:dIneq}
    d \leq \frac{\lambda}{1 + \abs{\sin(\max\{\theta+\Delta\})}},
\end{equation}
where $\max\{\theta+\Delta\}\in [-\pi/2, \pi/2]$ is the maximum steering angle with respect to vector normal to the midpoint of the antenna array.\vspace{-0.5cm}\end{remark}
\begin{table}[h!]
\centering
\caption{Coverage range of a steering vector ($\Delta$ bounds) in degrees: analytical vs numerical results.}
\label{tab:numerVsAnalytic}
\begin{tabular}{|c|c|c|c|c|}
\hline
\textbf{$\theta$} & $\mathrm{L}\Delta_{\theta}$(num.) & $\mathrm{L}\Delta_{\theta}$ & $\mathrm{U}\Delta_{\theta}$(num.)& $\mathrm{U}\Delta_{\theta}$\\
\hline
$0$ & -6.4013 & -6.3578 & 6.4013 & 6.3578\\
$15$ & -6.5279 & -6.4842 & 6.7347 & 6.6882\\
$30$ & -7.1382 & -7.0913 & 7.6974 &7.6428\\
$45$ & -8.4434 & -8.3896 & 9.9447 & 9.8695\\
$60$ & -11.0153 & -10.9494 & 17.8272 & 17.6240\\
\hline
\end{tabular}
\end{table}

In Tab. \ref{tab:numerVsAnalytic}, we compare the results for the range \(\Delta\) obtained through the analytical solution with those obtained numerically, for the case of $N=8$, $d=\lambda/2$ and $\gamma=3$\,dB. The table shows that the difference is negligible. Therefore, the provided analytical solution can be considered effectively exact for moderate values of \(\gamma_{\mathrm{f}}.\)

Equation \eqref{eq:ULAdeltaInterval} provides an analytical solution for the boundary of \( \Delta \) that ensures the array gain loss remains within \( \gamma \) dB, representing the angular coverage boundaries of a steering vector centered at \( \theta \). It is evident that the expression is a function of \( \theta \), which implies that the allowable deviation (coverage range) varies with the angle \( \theta \). Therefore, the beam sweeping step size is not constant for a given value of \( \gamma \); instead, it is angle-dependent. 

Fig. \ref{fig:devPlots} illustrates further \( D_{\theta}(\Delta) \) as a function of \( \Delta \), with results obtained for \( d = \lambda/2 \) and \( N = 8 \). The figure shows that, for a given maximum allowed degradation, the steering vector coverage is smaller as \( \theta \) approaches zero and increases as \( \theta \) diverges. Furthermore, for \( \gamma = 3 \) dB, the figure demonstrates that only a few steering vectors are needed to cover the range \( [0, 60^\circ] \).

\begin{figure}[t]
    \centering
        \includegraphics[width = .9\linewidth]{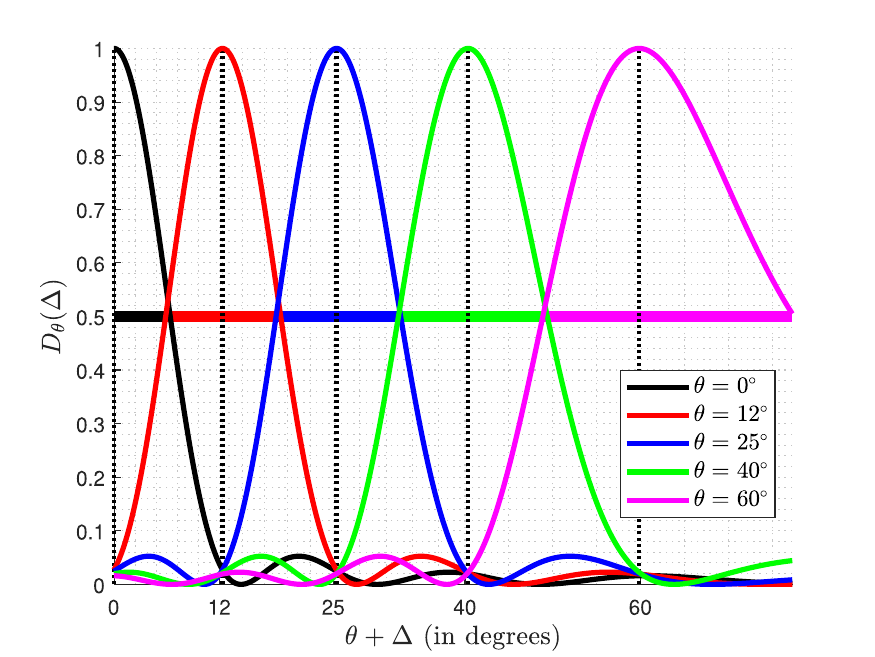}
        \caption{Achievable-to-maximum array gain as a function of the angle deviation $\Delta$ for different values of $\theta$.}\vspace{-0.3cm} 
        \label{fig:devPlots}
\end{figure}
\subsection{The case of an URA} \label{DevFromMaxGainForURA}
Consider an URA with a strong ray arriving from a DoA defined through an azimuth-elevation angle pair $(\theta_1,\theta_2)$. The array manifold vector of an $N_1 \times N_2$ URA is
\begin{equation} 
\begin{aligned}
\mathbf{a}(\theta_1, \theta_2)&=[a_n(\theta_1,\theta_2)]_{n=1..N_1\times N_2}\\&=\beta\mathbf{v}_{1}(\theta_1, \theta_2) \otimes \mathbf{v}_{2}(\theta_1, \theta_2),
\end{aligned}
\end{equation}
where $\mathbf{v}_{1}$ and $\mathbf{v}_{2}$ are the vectors on the horizontal and vertical directions, respectively, with expressions
\begin{align} \label{eq:steeringVecs}
&\mathbf{v}_{1}(\theta_1, \theta_2)=\nonumber\\&\qquad \left[1, e^{j \frac{2\pi}{\lambda}d_{1} \sin \theta_2 \cos \theta_1} \cdots, e^{j \left(N_1-1\right) \frac{2\pi}{\lambda}d_{1} \sin \theta_2 \cos \theta_1}\right]^{T},\\
&\mathbf{v}_{2}(\theta_1, \theta_2)=\nonumber\\&\qquad\left[1, e^{j \frac{2\pi}{\lambda}d_{2} \sin \theta_2 \sin \theta_1} \cdots, e^{j\left(N_2-1\right) \frac{2\pi}{\lambda}d_{2} \sin \theta_2 \sin \theta_1}\right]^{T}. 
\end{align}
Here, $d_1$ and $d_2$ are the horizontal and vertical distances between neighboring antennas, respectively.
Similar to the case of a URA, the antenna array gain is written as follows:
\begin{equation}\label{eq:URAgain}
\begin{aligned}
G(\mathbf{w},\theta_1, \theta_2)&= \abs{\mathbf{a}(\theta_1, \theta_2)\mathbf{w}}^2 \\&= \frac{1}{N_1N_2}\abs{\sum_{n=1}^{N_1N_2}a_n(\theta_1, \theta_2)e^{-j\phi_n}}^2. 
\end{aligned}
\end{equation}
Considering an infinitely large beam sweeping codebook with an infinitesimally small step, there will be a beam with a direction that perfectly aligns with the DoA of the strongest ray. The maximum array gain is achieved by setting the phase of \( e^{-j\phi_n} \) to be equal to the phase of the complex conjugate of \( a_n(\theta_1, \theta_2) \), resulting in an array gain of \( \beta^2 N_1 N_2 \). Let \( \mathbf{w}_{\theta_1, \theta_2} \) represent the optimal steering vector for the azimuth-elevation angle pair \( (\theta_1, \theta_2) \), with elements defined as before. Now, consider applying the steering vector \( \mathbf{w}_{\theta_1, \theta_2} \) to a signal with an incident wave with slightly different angle \( (\theta_1 + \Delta_1, \theta_2 + \Delta_2) \). In this scenario, the array gain is degraded by a factor given by:

\begin{equation}\label{eq:URAdevgain}
\begin{aligned}
    D_{\theta_1,\theta_2}(\Delta_1,\Delta_2)
&= \frac{G(\theta_1 + \Delta_1, \theta_2 + \Delta_2, \mathbf{w}_{\theta_1,\theta_2})}{G(\theta_1, \theta_2, \mathbf{w}_{\theta_1,\theta_2})} \\&= \frac{G(\theta_1 + \Delta_1, \theta_2 + \Delta_2, \mathbf{w}_{\theta_1,\theta_2})}{N_1N_2\abs{\beta}^2}\\
   & = \frac{1}{(N_1N_2)^2} \abs{\sum_{m=1}^{N_2}\sum_{n=1}^{N_1}e^{j(m-1)z_2} e^{j(n-1)z_1}}^2,
\end{aligned}
\end{equation}
where
\begin{align}\label{eq:URAdev}
z_1 = \frac{2\pi}{\lambda}d_1[\sin(\theta_1+\Delta_1)\cos(\theta_2+\Delta_2) - \sin(\theta_1)\cos(\theta_2)], \\
z_2 = \frac{2\pi}{\lambda}d_2[\sin(\theta_1+\Delta_1)\sin(\theta_2+\Delta_2) - \sin(\theta_1)\sin(\theta_2)].
\end{align}

Using the geometric sum formula and the respective trigonometric identity, (\ref{eq:URAdevgain}) can be simplified to 
\begin{equation}\label{eq:URAdevSimp}
\begin{aligned}    
    D_{\theta_1,\theta_2}(\Delta_1,\Delta_2) = \frac{1}{(N_1N_2)^2}\prod_{i=1}^2\abs{\frac{1 - e^{jN_iz_i}}{1-e^{jz_i}}}^2 \\ =  \frac{1}{(N_1N_2)^2}\prod_{i=1}^2\frac{1-\cos(N_iz_i)}{1-\cos(z_i)} \\ =  \left(\prod_{i=1}^2\frac{1}{N_i}\frac{\sin(N_iz_i/2)}{\sin(z_i/2)}\right)^2.  
\end{aligned}
\end{equation}
Constraining the maximum degradation to $\gamma_{\mathrm f}$, we have 

\begin{equation}\label{eq:URAineq2}
    D_{\theta_1,\theta_2}(\Delta_1,\Delta_2) \geq \frac{1}{\gamma_{\mathrm f}} \Longleftrightarrow \abs{\prod_{i=1}^2\frac{1}{N_i}\frac{\sin(N_iz_i/2)}{\sin(z_i/2)}} \geq \frac {1}{\sqrt{\gamma_{\mathrm f}}}.
\end{equation}
Different from the case of URA, considering the expressions of $z_i$ as a function of the azimuth-elevation coordinate system leads to intractable solution for the steering vector coverage range. Instead, we express the incident angle and the array gain degradation as functions of the angles relative to the x-axis and y-axis. Fig. \ref{fig:angleSys} provides a visual comparison between the azimuth-elevation angle system and the one specified by the rotation angles with respect to the x and y axes, denoted by $\theta_x$ and $\theta_y$, respectively. Note that there exists a one-to-one mapping system, ensuring that this has no impact on identifying the subset of steering vectors. Given values in one coordinate space, the corresponding coordinates in the other space can be uniquely determined. Specifically, we have:
\begin{align}
    x = \sin(\theta_2)\cos(\theta_1), \\
    y = \sin(\theta_2)\sin(\theta_1).
    \end{align}
for the Cartesian coordinates $x,y$ and the azimuth-elevation angles $\theta_1,\theta_2$. Thus, we get the relation
\begin{align} \label{eq:elevxyTrans}
    x = \sin(\theta_y) = \sin(\theta_2)\cos(\theta_1), \\
    y = \sin(\theta_x) = \sin(\theta_2)\sin(\theta_1).
    \end{align} 
between the pairs $(\theta_1,\theta_2)$ and $(\theta_x,\theta_y)$. Given that the visibility of a ULA lies within the range \([-\pi/2, \pi/2]\), the mapping described above provides a valid and unique one-to-one correspondence, i.e., each pair \((\theta_1, \theta_2)\) uniquely maps to a pair \((\theta_x, \theta_y)\) and vice versa.

\begin{figure}[h!]
    \centering
    \includegraphics[width=.9\linewidth]{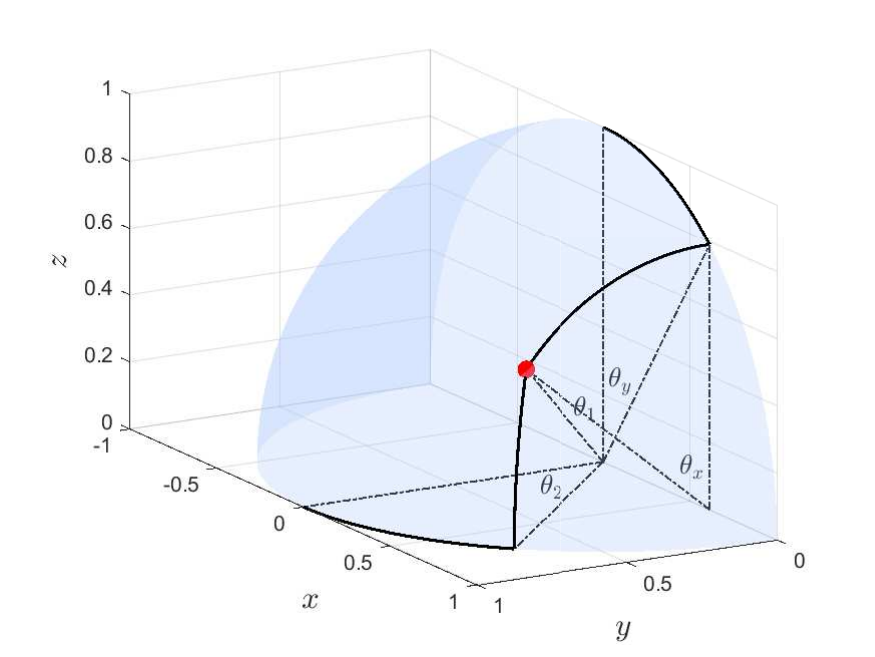}
    \caption{Azimuthal-elevation vs Cartesian coordinates systems}
    \label{fig:angleSys}
\end{figure}

Let \(\Delta_x\) and \(\Delta_y\) represent the angular deviations with respect to the x-axis and y-axis, respectively. Accordingly, we have:
\begin{align} \label{eq:steeringVecsElevxy}
&\mathbf{v}_{1}(\theta_x+\Delta_x) =\nonumber\\& \left[1, e^{j \frac{2\pi}{\lambda}d_{1} \sin(\theta_x+\Delta_x)} \cdots, e^{j \left(N_1-1\right) \frac{2\pi}{\lambda}d_{1} \sin(\theta_x+\Delta_x)}\right]^{T},\\
&\mathbf{v}_{2}(\theta_y+\Delta_y) =\nonumber\\& \left[1, e^{j \frac{2\pi}{\lambda}d_{2} \sin(\theta_y+\Delta_y)} \cdots, e^{j \left(N_2-1\right) \frac{2\pi}{\lambda}d_{2} \sin(\theta_y+\Delta_y)}\right]^{T}.
\end{align}
Moreover, the expression for the ratio of the achievable array gain to the maximum array gain, considering an angular deviation of \((\Delta_x, \Delta_y)\), is given by:
\begin{equation}\label{DevElevxy}
    \bar{D}_{\theta_x,\theta_y}(\Delta_x,\Delta_y) =  \left(\prod_{i=1}^2\frac{1}{N_i}\frac{\sin(N_i\bar{z_i}/2)}{\sin(\bar{z_i}/2)}\right)^2, 
\end{equation}
where 
\begin{align}\label{eq:URAzElevxy}
\bar{z}_1 = \frac{2\pi}{\lambda}d_1[\sin(\theta_x+\Delta_x) - \sin(\theta_x)], \\
\bar{z}_2 = \frac{2\pi}{\lambda}d_2[\sin(\theta_y+\Delta_y) - \sin(\theta_y)].
\end{align}
 Analogous to (\ref{eq:URAineq2}), we can write
\begin{equation}\label{eq:URAineq2Elevxy}
\bar{D}_{\theta_x,\theta_y}(\Delta_x,\Delta_y) \geq \frac{1}{\gamma_{\mathrm f}} \Longleftrightarrow \prod_{i=1}^2\frac{1}{N_i}\frac{\abs{\sin(N_i\bar{z_i}/2)}}{\abs{\sin(\bar{z_i}/2)}} \geq \frac {1}{\sqrt{\gamma
_{\mathrm f}}}.
\end{equation}

\begin{lemma}\label{lem:2}
Considering $\gamma_{\mathrm f} > 1$, the following inequality holds:

\begin{equation}\label{eq:lem2ineq}
\prod_{i=1,2} \frac{1}{N_i} \frac{\abs{\sin(N_i \bar{x}_i / 2)}}{\abs{\sin(\bar{z}_i / 2)}} \geq \frac{1}{\sqrt{\gamma_{\mathrm f}}},
\end{equation}
when $ \bar{z}_i\in[-\frac{2}{N_i} \alpha^{*}, \frac{2}{N_i} \alpha^{*}], \, i = \{1,2\} $. Here, $\alpha^{*}$ is the unique root of the equation:
$$\frac{\sin(\alpha)}{\alpha} = \frac{1}{\sqrt[4]{\gamma_{\mathrm f}}},$$
within the range $[0, \pi]$.
\end{lemma}

\begin{proof}
We have $$\abs{\sin\left(\frac{\bar{z}_i}{2}\right)}\leq \abs{\frac{\bar{z}_i}{2}}.$$ Thus, the range of values of \( \bar{z}_i \) in which
\begin{equation}\label{eq:inealpha}
\frac{\abs{\sin(N_i\bar{z_i}/2)}}{\abs{N_i\bar{z_i}/2}} \geq \frac {1}{\sqrt[4]{\gamma
_{\mathrm f}}},\,\, i={1,2}\end{equation} holds, also satisfies \eqref{eq:URAineq2Elevxy}.
Since \(\gamma_{\mathrm{f}}\) is strictly greater than 1, the same applies to \(\sqrt[4]{\gamma_{\mathrm{f}}}\). Now, consider the inequality:
\[
\frac{\abs{\sin(\alpha)}}{\abs{\alpha}} \geq \frac{1}{\sqrt[4]{\gamma_{\mathrm{f}}}}.
\]
The first derivative of the function on the left-hand side is given by:
\[
\frac{\alpha \cos(\alpha) - \sin(\alpha)}{\alpha^2}.
\]
This derivative is negative in the range \((0, \pi]\), indicating that the function is strictly decreasing within this interval. Moreover, the function takes the value 1 when \(\alpha = 0\) and decreases to 0 as \(\alpha\) approaches \(\pi\). Together, these imply that the function intersects the constant line defined by \(0 < \frac{1}{\sqrt[4]{\gamma_{\mathrm{f}}}} < 1\) exactly once within the range \([0, \pi]\). Thus, there exists a unique \(\alpha^{*} \in [0, \pi]\) such that:
\[
\frac{\abs{\sin(\alpha)}}{\abs{\alpha}} = \frac{1}{\sqrt[4]{\gamma_{\mathrm{f}}}}.
\]
 Furthermore, for \(\alpha \in [-\alpha^{*}, \alpha^{*}]\), the inequality in \eqref{eq:URAineq2Elevxy} holds. Given that the function is bijective and following the same reasoning, it can be concluded that the inequality also holds within \([- \alpha^{*}, 0]\), resulting a total range \([- \alpha^{*}, \alpha^{*}]\) where the inequality is satisfied. 

Substituting $\alpha$ by $N_i\bar{z}_i/2$ implies that the constraint in \eqref{eq:lem2ineq} is satisfied when $$ \bar{z}_i\in[-\frac{2}{N_i} \alpha^{*}, \frac{2}{N_i} \alpha^{*}], \, i = 1,2.$$
\end{proof}

Considering the expression of $\bar{z}_i$ and the results provided in Lem. \ref{lem:2}, we obtain
\begin{align}\label{eq:URAzIneqElevxyExp} 
   \frac{-\alpha^{*}\lambda}{\pi d_1 N_1} + \sin(\theta_x) \leq \sin(\theta_x + \Delta_x) \leq \frac{\alpha^{*}\lambda}{\pi d_1 N_1} + \sin(\theta_x) \\
   \frac{-\alpha^{*}\lambda}{\pi d_2 N_2} + \sin(\theta_y) \leq \sin(\theta_y + \Delta_y) \leq \frac{\alpha^{*}\lambda}{\pi d_2 N_2} + \sin(\theta_y).
\end{align}
Recall that the visibility range of an ULA is within $[-\pi/2,\pi/2]$, and hence the $\sin(\cdot)$ function is invertible. This gives
\begin{equation}\label{eq:URAzIneqElevxFin}
\begin{aligned} 
   &\arcsin\left(\frac{-\alpha^{*}\lambda}{\pi d_1 N_1} + \sin(\theta_x)\right) - \theta_x \leq \Delta_x \\&\qquad\qquad\qquad\leq \arcsin\left(\frac{\alpha^{*}\lambda}{\pi d_1 N_1} + \sin(\theta_x)\right) - \theta_x, 
\end{aligned}
\end{equation}
\begin{equation}\label{eq:URAzIneqElevyFin}
\begin{aligned} 
   &\arcsin\left(\frac{-\alpha^{*}\lambda}{\pi d_2 N_2} + \sin(\theta_y)\right) - \theta_y \leq \Delta_y  \\&\qquad\qquad\qquad\leq \arcsin\left(\frac{\alpha^{*}\lambda}{\pi d_2 N_2} + \sin(\theta_y)\right) - \theta_y.
\end{aligned}
\end{equation}

\section{Codebook refinement as a set cover problem} \label{sec:minCodebookSCP}
In the previous section, we derived the expression for the coverage range (angle deviation range) of a steering vector that ensures the array gain loss remains within $\gamma$\,dB from the maximum gain. Recall that this range is not constant; rather, it depends on the angle $\theta$ of the steering vector. Specifically, for angles near zero, the array gain decays more rapidly with increasing error compared to larger $\theta$ values. This behavior is illustrated in Fig. \ref{fig:devPlots}, which provides a visual example. 
\par The dependence of the coverage range on the steering angle indicates that optimal beam sweeping should not be performed with equal steps. Therefore, it is crucial to carefully select the steering vectors to ensure complete coverage of the URA/ULA visibility range using the minimum number of vectors. Building on the derived expression for the coverage range of a steering vector, we formulate the problem of constructing a codebook with reduced size denoted by $\zeta_R$. We then leverage results from the set cover problem to develop a near-optimal algorithmic solution with polynomial complexity. Specifically, we demonstrate that the problem converges to a special case of the set cover problem. It is important to note that the codebook refinement process is carried out offline, and during transmission (online), only the refined codebook is used. This approach is employed in the experimental results, where the refined codebook is utilized and its performance is compared to that of the entire/intial codebook with large size.

\subsection{The ULA Case}
Let \( \Phi = \{\phi_1, \phi_2, \dots, \phi_M\} \) denote the set of beam sweeping directions of the initial and large codebook (\( \mathcal{C} = \{\mathbf{w}_{\phi_1}, \mathbf{w}_{\phi_2}, \dots, \mathbf{w}_{\phi_M}\} \)), which spans the angular visibility range of the ULA, denoted by \( V = [\phi^{\min}, \phi^{\max}] \). The angles in the set are sorted in increasing order. A steering vector \( \mathbf{w}_{\phi_M} \) can cover a range \( V_{\mathbf{w}_m} = [-\mathrm{L}\Delta_{\phi_m}, \mathrm{U}\Delta_{\phi_m}] \), as defined in \eqref{eq:ULAdeltaInterval}, while maintaining the gap to the maximum within \( \gamma \) dB. The codebook refinement problem can thus be formulated as:
\begin{equation}
\label{eq:ULASCPgeneral}
\begin{aligned}
\underset{a_m, m=1..M}{\text{minimize}} \quad & \sum_{i=1}^M{a_i}\\
\textrm{subject to} \quad & \bigcup_{m=1..M}a_mV_{\mathbf{w}_m}\supseteq V\\
\quad& a_m\in\{0,1\},\,\, m=1..M 
\end{aligned}
\end{equation}
In the above minimization problem, the angular visibility of the ULA and the angular coverage of the steering vectors are continuous. By discretizing both through uniform sampling, we transform the problem into a set cover problem. Specifically, we discretize the visibility range \( V \) into a vector of binary values, where each index corresponds to a specific angle, and the value at each index is \( 1 \) if the corresponding angle is within the visibility range, and \( 0 \) otherwise. Similarly, the angular coverage of each steering vector is represented as a binary vector, where the value is \( 1 \) if the corresponding angle in \( V \) is covered by the steering vector and \( 0 \) otherwise. Let \( \tilde{V} \) and \( \tilde{V}_{\mathbf{w}_m} \) denote the discrete versions of \( V \) and \( V_{\mathbf{w}_m} \), respectively. This discretization transforms the problem into a set cover problem \cite{combinatorial2021korte}, which can be formulated as follows:

\begin{equation}
\label{eq:ULASCP}
\begin{aligned}
\underset{a_m, m=1..M}{\text{minimize}} \quad & \sum_{i=1}^M{a_i}\\
\textrm{subject to} \quad & \sum_{m=1..M} a_m \tilde{V}_{\mathbf{w}_m} = \tilde{V}\\
\quad & a_m \in \{0, 1\}, \quad m = 1..M
\end{aligned}
\end{equation}
Although the greedy algorithm is a well-known approach for solving set cover problems, the problem presented here has a distinctive property: each steering vector’s coverage forms a contiguous block of ones, with zeros outside the covered range. This feature allows for the reduction of the search space. The optimization process begins by sorting the steering vectors such that the vector that covers the smallest uncovered element in \( \tilde{V} \) is placed first, followed by the others in increasing order of their starting coverage element. At each step, the goal is to cover the first uncovered element of the remaining set by selecting the vector that extends the coverage as far as possible, i.e., the vector that covers the first uncovered element and provides the largest coverage. After each vector is selected, the covered elements are marked, and the process repeats for the remaining uncovered elements. This greedy approach ensures that each vector is chosen to maximize coverage while minimizing redundancy, thus efficiently solving the set cover problem. The pseudo-code for this algorithm is shown in Alg. \ref{alg:ulaAlg}. After solving the optimization problem and obtaining the values of \( a_m \), the refined codebook is reconstructed as: 
$
\zeta_R = \bigcup_{m=1}^{M} a_m \mathbf{w}_{m}.
$

\begin{algorithm}
\caption{Greedy Set Cover with Contiguous Coverage}
\label{alg:ulaAlg}
\begin{algorithmic}[1]
\STATE Initialize:\\
\(\tilde{V_1}\) = $\{\phi_1,\phi_2, \dots, \phi_M\}$ where $\phi_1\leq\phi_2\leq\dots\leq\phi_M$.\\
\(\tilde{V}_\text{uncovered} \) be the set of uncovered elements.\\
$\zeta_R \gets \varnothing$.\\
$i\gets 1$

\WHILE{\(\tilde{V}_\text{uncovered} \neq \emptyset\)}
    \STATE $i \gets \min(\{i:\phi_i\in\tilde{V}_{\text{uncovered}}\})$
    \STATE $N_{m,i} = \max(\{card| \{\phi_i,\phi_{i+1},\dots\}|:V_{\mathbf{w}_m} \supset \{\phi_i,\phi_{i+1},\dots\}\})$.
    \STATE $\zeta_R \gets \zeta_R \cup \{\mathbf{w}_{\hat{m}}\}$
    \STATE $\tilde{V}_\text{uncovered}$ $\gets \tilde{V}_\text{uncovered} \setminus \{\phi_j, \phi_{j+1}, \dots, \phi_{N_{\hat{m},i}}\}$
\ENDWHILE
\RETURN $\zeta_R$
\end{algorithmic}
\end{algorithm}

\subsection{The URA Case} \label{sec:SCPURAcase}
Considering the case of an URA, we define two sets corresponding to the azimuthal and elevation beam sweeping directions that defines the initial and large codebook: \( \Phi_1 = \{\phi_{1,1}, \phi_{1,2}, \dots, \phi_{1,M}\} \) and \( \Phi_2 = \{\phi_{2,1}, \phi_{2,2}, \dots, \phi_{2,K}\} \). These sets span the angular visibility range of the URA, denoted by \( V = [\phi^{\min}_1, \phi^{\max}_1] \times [\phi^{\min}_2, \phi^{\max}_2] \), where the angles in each set are sorted in increasing order. The steering vector corresponding to an azimuth-elevation direction pair \( (\phi_{1,i}, \phi_{2,j}) \) is denoted by \( \mathbf{w}_{i,j} \), which achieves the maximum array gain when the incident wave arrives from this direction. A steering vector \( \mathbf{w}_{i,j} \) can cover a range \( V_{\mathbf{w}_{i,j}} = [\mathrm{L}\Delta_{\phi_{1,i}}, \mathrm{U}\Delta_{\phi_{1,i}}] \times [\mathrm{L}\Delta_{\phi_{2,j}}, \mathrm{U}\Delta_{\phi_{2,j}}] \), as defined in \eqref{eq:URAzIneqElevxFin} and \eqref{eq:URAzIneqElevyFin}, while maintaining the gap to the maximum within \( \gamma \) dB. The codebook refinement problem can thus be formulated as:

\begin{equation}
\label{eq:URASCPgeneral}
\begin{aligned}
\underset{a_{i,j}}{\text{minimize}} \quad & \sum_{i=1}^M\sum_{j=1}^K{a_{i,j}}\\
\textrm{subject to} \quad & \bigcup_{i,j}a_{i,j}V_{\mathbf{w}_{i,j}}\supseteq t V\\
\quad& a_{i,j}\in\{0,1\}, i=1..M 
\, and \,\, j=1..K
\end{aligned}
\end{equation}

Similar to the case of ULA, by sampling the visibility range and the coverage range of the steering vectors, we obtain

\begin{equation}
\label{eq:URASCP}
\begin{aligned}
\underset{a_{i,j}}{\text{minimize}} \quad & \sum_{i=1}^M\sum_{j=1}^K{a_{i,j}}\\
\textrm{subject to} \quad & \sum_{i,j}a_{i,j}\tilde{V}_{\mathbf{w}_{i,j}}\supseteq \tilde{V}\\
\quad& a_{i,j}\in\{0,1\}, i=1..M 
\, and\,\, j=1..K
\end{aligned}
\end{equation}
Here, $\tilde{V}$ and $\tilde{V}_{\mathbf{w}_{i,j}}$ are the discrete versions of $V$ and $V_{\mathbf{w}_{i,j}}$, respectively. Since the coverage space is two-dimensional, the problem corresponds to a geometric set cover problem. The problem shares the same property with the ULA case, where the points covered by a steering vector are contiguous. This property can improve the convergence time of the greedy algorithm. While the algorithm differs from the ULA case, we propose a modified version. In particular, at each iteration, the algorithm selects the steering vector that covers either the minimum azimuthal or minimum elevation angles in the uncovered set while maximizing the total cardinality of the covered set in both azimuthal and elevation dimensions. For the sake of the clarity, the pseudo-code is provided in Alg. \ref{alg:uraAlg}.

In the algorithm's initialization step, the azimuth and elevation angular ranges of the URA are provided as input and treated as an uncovered set. The total set of angle pairs to be covered is then defined as the Cartesian product of the specified angular ranges. The refined codebook, \( \zeta_R \), is initially empty. The algorithm iteratively selects the steering vector that provides the maximum coverage, prioritizing those that cover the angle pair with either the smallest uncovered azimuth angle or the smallest uncovered elevation angle. Once a steering vector is selected, it is added to $\zeta_R$, and its coverage is removed from the set of uncovered angle pairs. This process continues until all angle pairs are covered.

\begin{algorithm}[t]
\caption{Greedy Set Cover with Contiguous Two Dimensional Coverage}
\label{alg:uraAlg}
\begin{algorithmic}[1]
\STATE Initialize:\\
\(\tilde{V}_1\) = $\{\phi_{1,1},\phi_{1,2}, \dots, \phi_{1,M}\}$ where $\phi_{1,1}\leq\phi_{1,2}\leq\dots\leq\phi_{1,M}$.\\
\(\tilde{V}_2\) = $\{\phi_{2,1},\phi_{2,2}, \dots, \phi_{2,K}\}$ where $\phi_{2,1}\leq\phi_{2,2}\leq\dots\leq\phi_{2,K}$.\\
Let \(\tilde{V}_\text{uncovered} = \tilde{V}_1 \times \tilde{V}_2 \) be the set of uncovered elements.\\
$\zeta_R \gets \varnothing$.\\
$i\gets 1$

\WHILE{\(\tilde{V}_\text{uncovered} \neq \emptyset\)}
    \STATE $i \gets \min(\{i:(\phi_{1,i},\phi_{2,k})\in\tilde{V}_{\text{uncovered}}\})$
    \STATE $j \gets \min(\{j:(\phi_{1,i},\phi_{2,j})\in\tilde{V}_{\text{uncovered}}\})$
    \STATE $(l,m) \gets \underset{l,m}\argmin(\{|\tilde{V}_\text{uncovered}\setminus \tilde{V}_{\mathbf{w}_{l,m}}|:\tilde{V}_{\mathbf{w}_{l,m}} \supset\{(\phi_{1,i},\phi_{2,j})\})$ 
    \STATE $\zeta_R \gets \zeta_R \cup \{\mathbf{w}_{l,m}\}$
    \STATE $\tilde{V}_\text{uncovered}$ $\gets \tilde{V}_\text{uncovered} \setminus \tilde{V}_{\mathbf{w}_{l,m}}$
\ENDWHILE
\RETURN $\zeta_R$
\end{algorithmic}
\end{algorithm}

\section{Experimental Results}
\subsection{Experimental Setup}\label{sec:expSetup}
The experimental results encompass two environments. The first is an anechoic chamber, where both LoS and NLoS scenarios are considered. The second environment involves an indoor office/laboratory setting, where the receiver is mounted on a mobile cart that is randomly positioned and oriented. In both environments, the receiver's orientation is random, with no prior knowledge of its orientation or position. More details on the experiment environments, including the dimensions are provided in Fig. \ref{anechoic_chamber}. The receiver is solely endowed with the refined version of the intial codebook \(\zeta\), denoted by $\zeta_R$. For comparison purposes, the maximum array gain achieved through channel estimation is used as a benchmark. Moreover, we compare the performance with hierarchical beamsteering, in which the same assumptions hold including the absence of prior knowledge about the location or orientation. 

\begin{figure*}[]
     \centering
     \subfigure[Illustration of anechoic chamber setup]{\includegraphics[width=3in]{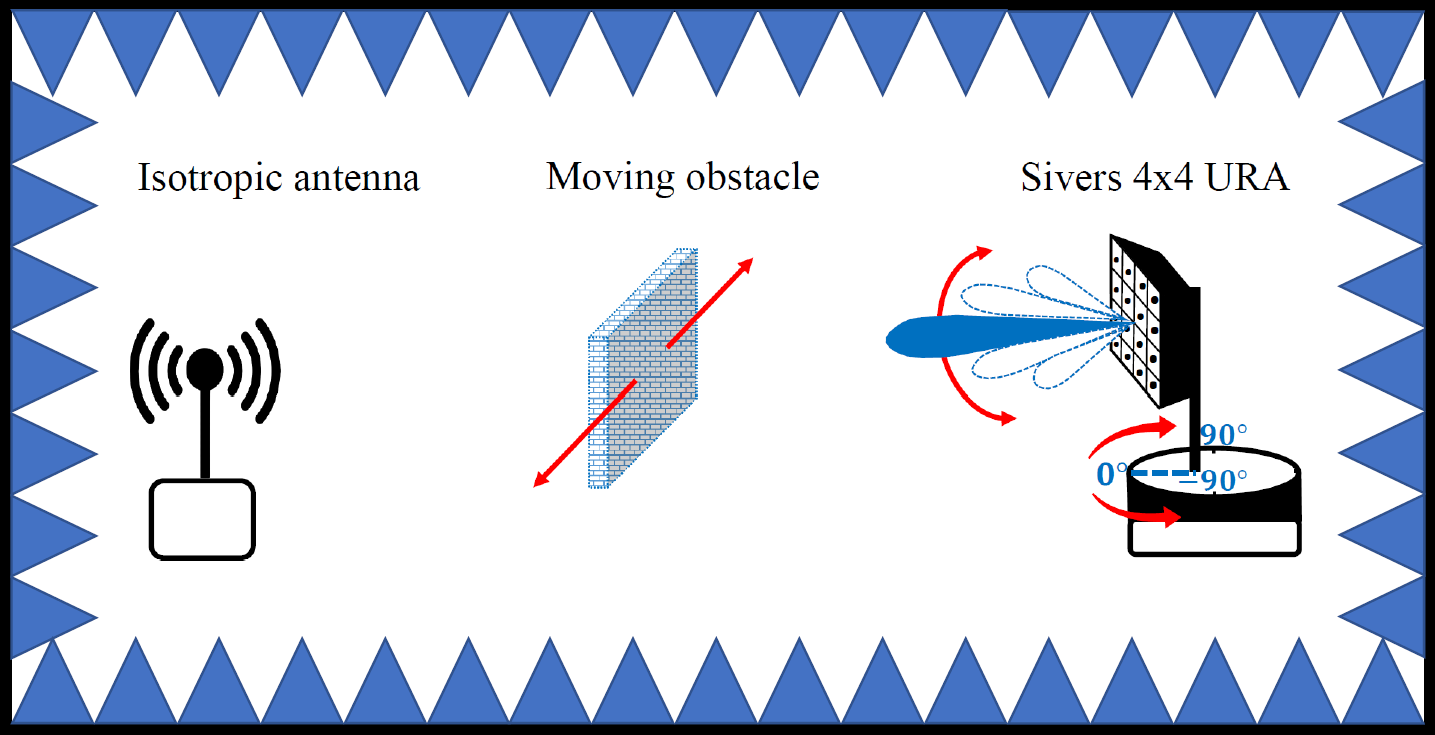} \label{fig:illustration}}
\hfill
    \subfigure[Anechoic chamber real setup]{\includegraphics[width=3in, height=1.52in]{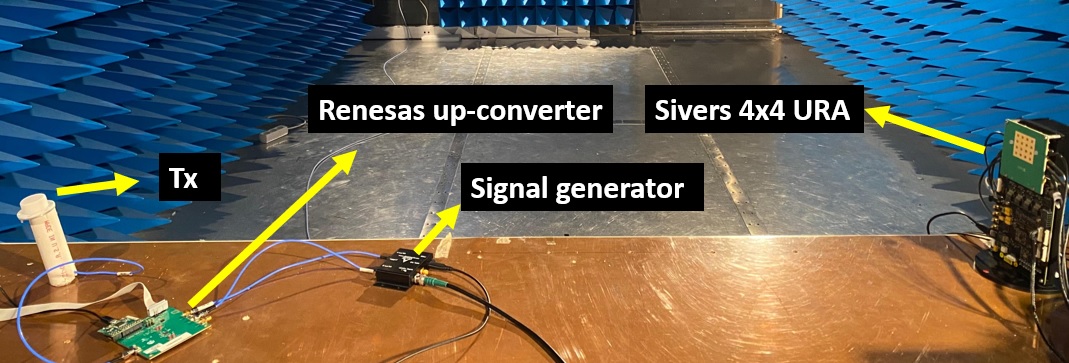} \label{fig:illustration}}
\hfill
     \subfigure[Illustration of lab/office environment]{\includegraphics[width=3in]{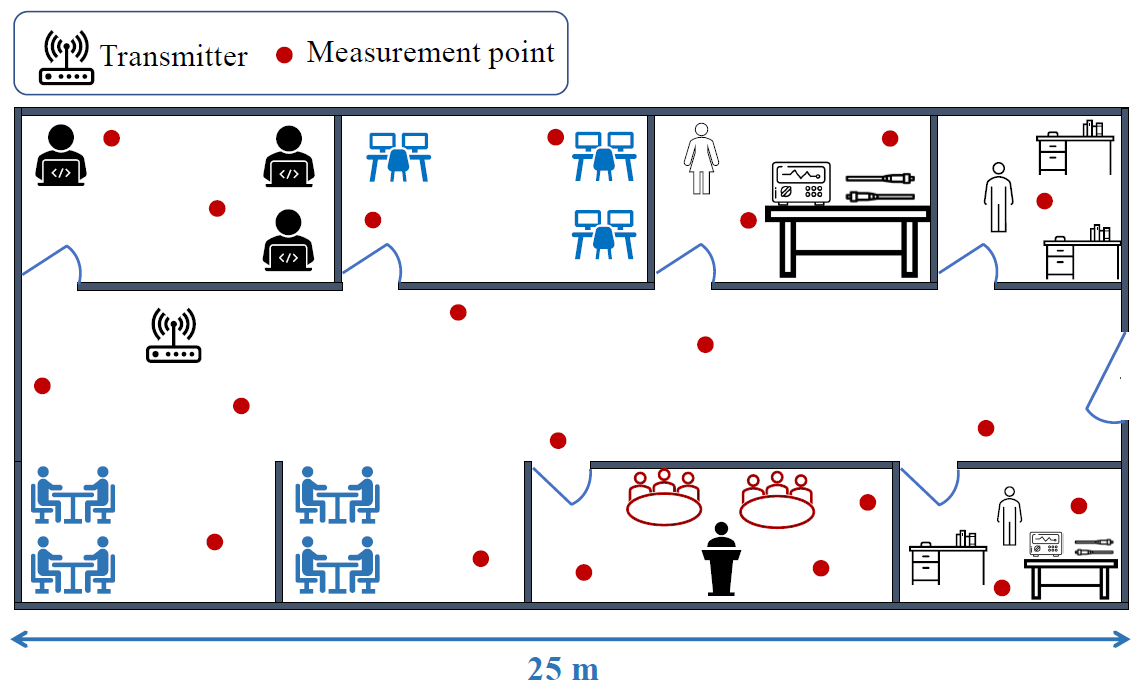}\label{fig:fixedAnechoicSetup}}
\hfill
     \subfigure[Lab/office environment]{\includegraphics[width=3in]{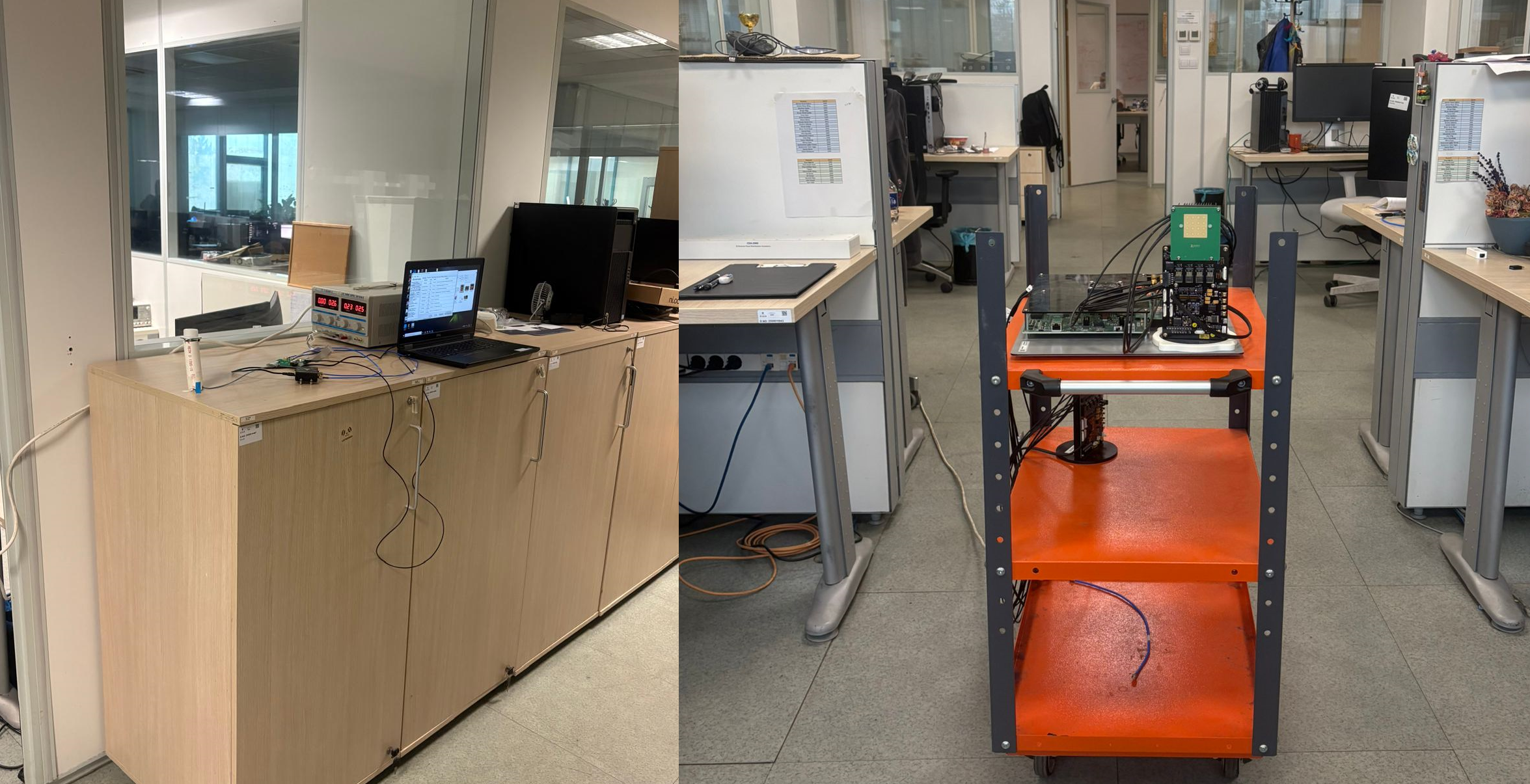} \label{fig:indoorOfficeEnv}}
     \caption{{Experiment environments}}
     \label{anechoic_chamber}
\end{figure*}
\begin{figure*}[h!]
     \centering
\subfigure[Achievable to the maximum array gain for $\gamma =3$ dB]
{\includegraphics[width=3.2in]{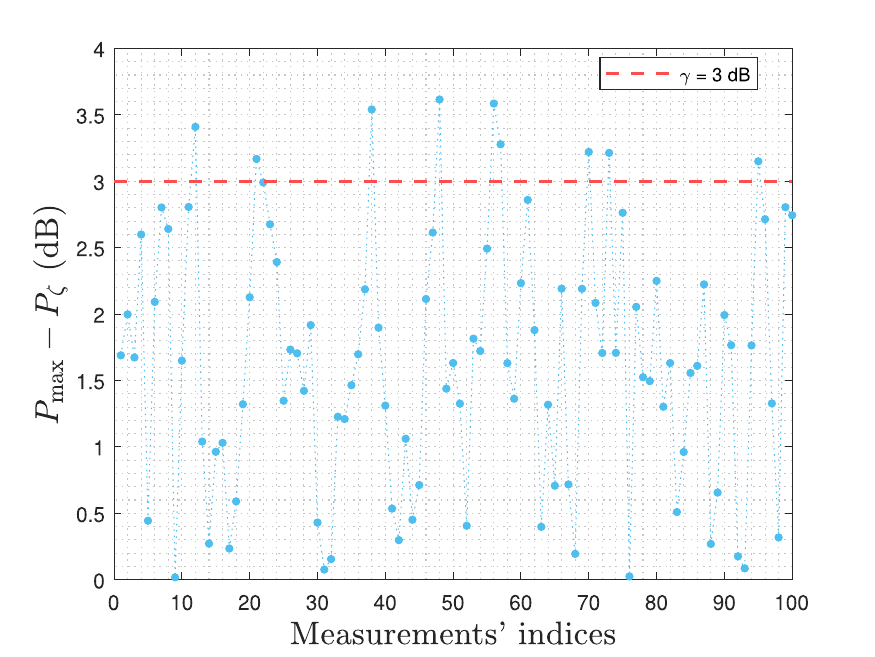} \label{fig:fig1}}
\subfigure[Achievable to the maximum array gain for $\gamma =2$ dB]
{\includegraphics[width=3.2in]{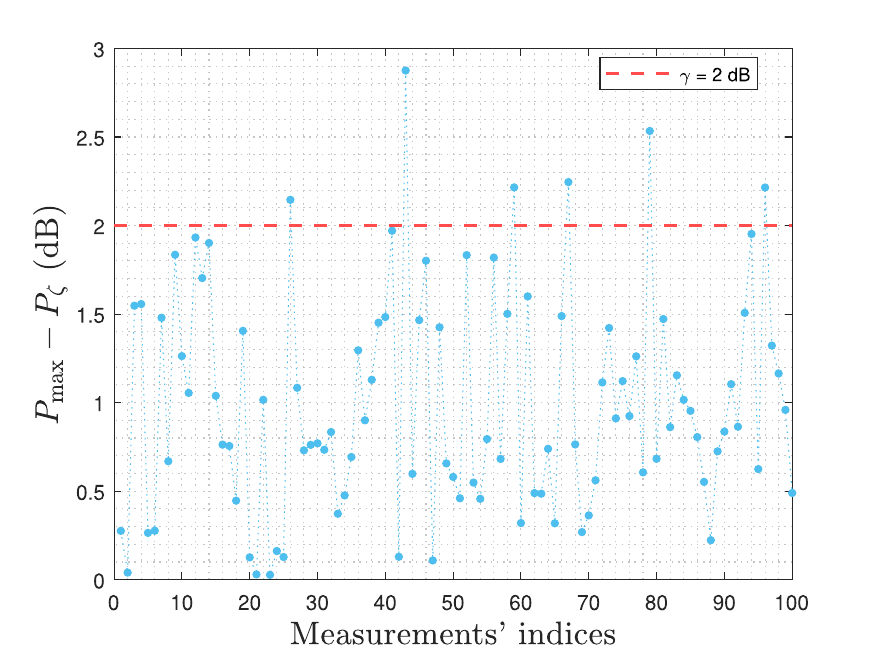} \label{fig:fig2}}
\caption{Achievable vs. Maximum Array Gain: The achievable gain is obtained using the refined codebook \( \zeta_R \), while the maximum gain is determined by scanning the entire initial codebook \( \zeta \).}\label{fig:indexdB}
\vspace{-0.5cm}
\end{figure*}

We use $4 \times 4$ Sivers URA receiver module of EVK02004 \cite{evk02004}. Each of the array's antennas is equipped with a 10-bit phase shifter, allowing for $1024$ possible settings per antenna. The case of ULA is obtained by activating only one row of antennas (four antennas) from the URA. In such a case, the URA and ULA codebook sizes, before refinement, are $1024^{16}$ and $1024^4$, respectively. The transmitter is equipped with one single-antenna. The antennas are arranged with a spacing of 5.15 mm between them. The Sivers URA orientation is random and changes over time. The measurements were taken from various positions, as shown by Fig. \ref{fig:indoorOfficeEnv}. Moreover, barriers with different dimensions, materials, and shapes were are randomly placed between the transmitter and the receiver during the experiment. For instance, the experiments are performed during office hours, when researchers are mobile and randomly located.

Signals are transmitted at a carrier frequency of 25.1\,GHz. The setup consists of a Renesas F5728 frequency multiplier \cite{f5728} and a mixer in conjunction with a signal generator. The signal generator produces two signals at frequencies of 3.1\,GHz and 5.5\,GHz. The 5.5\,GHz signal is supplied to the local oscillator input of the Renesas up-converter, which quadruples the frequency via its internal multiplier to generate a continuous wave at 22\,GHz. The output is then mixed with the 3.1\,GHz signal, resulting in a carrier frequency of 25.1\,GHz. The final signal is transmitted using an isotropic antenna.

A computer manages the entire system, overseeing both measurement collection and phase shifter configuration. In the validation phase that follows the refinement of the codebook, the device is rotated to a randomly selected angle. A beam is then steered in multiple directions according to the steering vectors in the refined codebook. The received power is measured, and the configuration that yields the array gain is determined through channel estimation and MRC. The gap between the maximum array gain and the gain achieved using the refined codebook afterward is measured and used to compute the array gain loss.

\subsection{Results}
\label{subsec:results}

We collected a total of 280 measurement points for validation, with 140 points for the case where \(\gamma = 3\) dB and 140 points for \(\gamma = 2\) dB. For each validation measurement point, we applied the following procedure: the multi-antenna receiver is oriented toward a random direction $\theta$ and the refined codebook $\mathcal{\zeta}_R$ is scanned to identify the steering vector that maximizes the received power $P_{\mathcal \zeta_R,\theta}$. One fifth of the collected measurements pertain the anechoic chamber environment, since the lab/office environment is more generic/realistic. The performance of the proposed approach is assessed by analyzing the gap between the achievable and maximum array gain $P_{\max,\theta}-P_{\mathcal \zeta_R, \theta}$. Maximum received power $P_{\max}$ is obtained through considering the initial and large codebook $\zeta$.

Recall that the analytical solution developed in this paper provides a subset of the angular range that can be covered by a steering vector, whereas the numerical solution offers the exact range. Tab. \ref{TabURA} pertains respectively, the ULA and URA case scenarios. We compare the refined code-book size derived through the theoretical bounds to the one derived numerically. 

\begin{table}[h!]
\centering
\hspace{-0.5cm}
\caption{Size of refined codebook for the case of ULA and ULA for different $\gamma$ values}
\label{tab:SizeY}
\begin{tabular}{|c|c|c|}
\hline
\textbf{$\gamma$ dB}  & \textbf{4-antenna ULA}  &\textbf{16-antenna URA}  \\
\hline
 1  & 11 & 20 \\
 2  & 9 & 17 \\
 3  & 6 & 14\\
 5  & 5 & 10\\
\hline
\end{tabular}\label{TabURA}
\end{table}

Tab. \ref{TabURA} indicates the codebook size obtained through the developed framework for different values of $\gamma$. In comparison with the initial codebbok size that is $1024^4$ for the case of 4-antenna ULA and $1024^{16}$ for the case of 16-antenna URA, there is a substantial reduction in the codebook size.  For example, the size of codebook is 14 when $\gamma$ is set to 3 dB.

The experimental results pertaining the refined beam sweeping are depicted in Figures \ref{fig:indexdB} and \ref{fig:emp cdf}. Fig. \ref{fig:indexdB} illustrates the gap between the maximum array gain and the achievable gain, \( P_{\max} - P_{\zeta} \), at each validation measurement point, where \( \gamma \) is set either to 3\,dB or 2\,dB. The figure demonstrates that the vast majority of test measurements fall within the intended gain range. Notably, for only 9 data points out of 140, the achievable gain falls below the pre-set threshold when $\gamma=3$dB. Since the experimental setup inherently includes noise, it is an integral part of the results, including the estimated array gain. The noisy results explain why some points exceed the threshold. Nonetheless, they remain within 4\,dB of the maximum achievable power. While the refined codebook clearly provides performance close to the maximum array gain, its size of 14 (as shown in Tab. \ref{TabURA}) represents a negligible fraction of the initial codebook size of \(1024^{16}\).

To further illustrate the efficiency of the proposed codebook generation technique and the likelihood of maintaining the gain difference within the intended range, we present the cumulative distribution function (CDF) of the achievable versus maximum gain in Fig. \ref{fig:emp cdf}. The results are shown for various selected values of $\gamma$. It is important to note that the pre-defined threshold results in different codebook sizes. The findings reveal that the constraint on achievable gain is met in over 90\% of cases, even in the presence of noise. The proposed method achieves gains closer to the maximum, even when utilizing a smaller codebook. For instance, with $\gamma$ set to 2\,dB, the codebook is reduced to 14 steering vectors. The refined codebook performance clearly outperforms the hierarchical approach.

\begin{figure}[t]
\centering
\includegraphics[width = 1\linewidth]{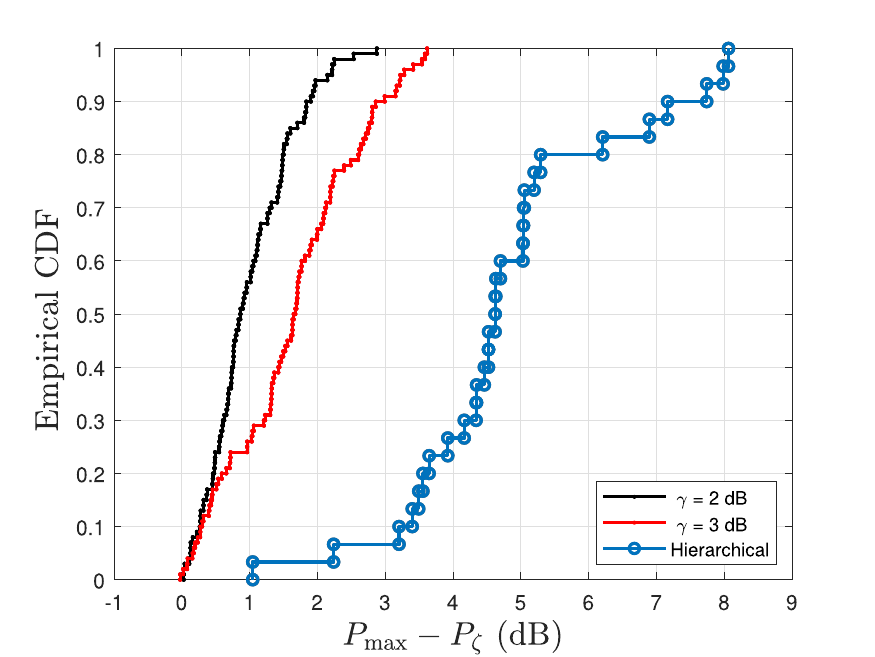}
\caption{Empirical CDF of the achievable to the maximum gain}
\label{fig:emp cdf} \vspace{-0.3cm}
\end{figure} 

\section{Conclusion}
The study conducted in this paper, together with the experimental validation, proves that a few steering vectors can substitute large codebooks, with reasonable degradation compared to the maximum array gain. For instance, one can reduce the codebook size from \(1024^{16}\) to 14 steering vectors while experiencing a maximum loss of 3 dB. Hence, a few configurations are enough to cover the entire visibility range of a URA/ULA while maintaining performance close to the maximum. This work makes no assumptions regarding the availability of location or direction information. Moreover, it eliminates the need for channel estimation, which typically requires a large overhead and may suffer from inaccuracies when the overhead is limited. This work provides evidence that it is possible to achieve an effective trade-off between computational complexity and array gain performance.

While the proposed technique has demonstrated effectiveness in beam steering, additional research is needed to develop a new method for rapid beam alignment, particularly in scenarios involving multiple-antenna transmitters and receivers. Even with a small codebook size at both the transmitter (e.g., \(N_t\)) and the receiver (e.g., \(N_r\)), the beam alignment process still involves searching across \(N_t \times N_r\) combinations, which can be quite extensive. Therefore, our future work will primarily focus on reducing the beam alignment search time. We plan to utilize machine learning techniques to design a cost-efficient beam alignment method.


\bibliographystyle{IEEEtran}
\bibliography{IEEEabrv,TWC1bib}

\end{document}